\documentclass[twocolumn]{revtex4-2}
\usepackage[utf8]{inputenc}
\usepackage{amsmath,amsfonts,amssymb,amsthm}
\usepackage{mathrsfs} 
\usepackage{tcolorbox}
\usepackage{bm,dsfont}
\usepackage{braket}
\usepackage{stmaryrd}
\usepackage{physics}
\usepackage{xcolor}
\usepackage{lipsum}
\usepackage[normalem]{ulem}
\usepackage[colorlinks=true, allcolors=blue
]{hyperref}
\usepackage{algpseudocode}

\usepackage{algorithm}

\usepackage{blindtext}
\usepackage{graphicx}
\usepackage{enumitem}
\usepackage{tikz}

\usetikzlibrary{automata,topaths,arrows,shapes,positioning,calc,decorations.pathreplacing}

\newtheorem{problem}{Problem}
\newtheorem{theorem}{Theorem}
\newtheorem{question}{Question}

\usepackage{color}

\newtheorem{lemma}{Lemma}

\begin{document}
\title{
Conditions for Advantageous Quantum Bitcoin Mining
}
\author{Robert R. Nerem}
\affiliation{Institute for Quantum Science and Technology,\\ University of Calgary, Alberta T2N 1N4, Canada  }
\email{riley.nerem@gmail.com}
\author{Daya R. Gaur}
\affiliation{Department of Mathematics and Computer Science, University of Lethbridge, Alberta T1K 3M4, Canada    }

\begin{abstract} 

Our aim is to determine conditions for quantum computing technology to give rise to security risks associated with quantum Bitcoin mining. Specifically, we determine the speed and energy efficiency a quantum computer needs to offer an advantage over classical mining. We analyze the setting in which the Bitcoin network is entirely classical except for a single quantum miner who has small hash rate compared to that of the network.
We develop a closed-form approximation for the probability that the quantum miner successfully mines a block, with this probability dependent on the number of Grover iterations the quantum miner applies before making a measurement. 
Next, we show that, for a quantum miner that is ``peaceful'',  this success probability is maximized if the quantum miner applies Grover iterations for 16 minutes before measuring, which is surprising as the network mines blocks every 10 minutes on average. Using this optimal  mining procedure, we show that the quantum miner outperforms a classical computer in efficiency (cost per block) if the condition $Q < Crb$ is satisfied, where $Q$ is the cost of a Grover iteration, $C$ is the cost of a classical hash, $r$ is the quantum miner's speed in Grover iterations per second, and $b$ is a factor that attains its maximum if the quantum miner uses our optimal mining procedure. This condition lays the foundation for determining when quantum mining, and the known security risks associated with it, will arise.

 \end{abstract}
\maketitle
\section{Introduction} \label{sec:introduction}
Bitcoin is a distributed digital currency deriving its security from cryptographic protocols \cite{Bitcoin}. The use of Bitcoin has increased dramatically in recent years, and blockchain, the backbone of the currency's security, has found application in a variety of areas \cite{BlockchainApp}.
The authenticity of Bitcoin's blockchain relies on a proof-of-work protocol in which ``miners" race to solve challenging computational problems. Although quantum computers promise faster solutions to these mining problems in theory, the practical benefits they provide and their impact on security remain uncertain.

We determine conditions for a quantum computer to outperform a classical computer at Bitcoin mining. This approach is in contrast with earlier work that instead ascertains if a quantum miner would be able to dominate the Bitcoin network \cite{Aggarwal}. We are motivated by the work of \cite{Sattath} which shows quantum mining poses a security risk even if a single quantum miner can not dominate the network.

\subsection{Quantum Bitcoin Mining}
 Bitcoin mining is the process through which blocks are added to the Bitcoin blockchain, a public ledger which records transactions.  In Bitcoin mining, miners attempt to partially invert a cryptographic hash function. That is, for a given hash function $f$, they search for a string $x$ such that $f(x)$ is less than some threshold $\tau$.  A solution to this problem yields a proof-of-work which allows the miner to add a block to the blockchain and receive a subsequent Bitcoin reward. Partial inversion of a hash function is equivalent to searching for a marked item in an unordered list of items (unstructured search) if $f$ is computationally inefficient to invert. Here, $x$ such that $f(x) < \tau$ are the marked items. Bitcoin mining is therefore an obvious potential application of Grover's quantum algorithm for unstructured search \cite{Grover1996}. Whereas classical algorithms for unstructured search find a marked item in $\Theta(D)$ queries to $f$, where $D := N/M$ is the ratio of total items to marked items, Grover's search can find a marked item in $\Theta(\sqrt D)$ queries \cite{Zalka1999, Boyer1998}. We refer to queries to a hash function $f$ by hashes. 
 
 Quantum and classical search algorithms differ in another important way. The optimal classical algorithm for unstructured search is the simple brute-force method of guessing and checking through the search space in some random order.   Contrarily, Grover's algorithm consists of applying some number of quantum operations called Grover iterations followed by a measurement, where each Grover iteration contains two hashes. Unlike the classical brute-force algorithm, which can yield a marked item after any hash, Grover's algorithm can only yield a marked item after the measurement step. Additionally, the more Grover iterations applied before measurement, the higher the probability that measurement yields a solution (assuming fewer than the optimal number of Grover iterations are applied).
 
This difference between classical and quantum search algorithms has key consequences to Bitcoin's proof-of-work. What makes Bitcoin sensitive to this difference is that miners race to be the first to find a marked item. It is important to win the race, as only a single winner reaps the Bitcoin reward for each block added. Once a block is added, a new search problem begins rendering progress on the previous problem irrelevant. For a classical miner the race setting does not change the optimal procedure from brute-force search. However, the difficulty in evaluating the effectiveness of quantum mining is that, for a quantum miner, determining the optimal procedure in the race setting requires answering the following 
question first asked in \cite{Sattath}.
 \begin{question}\label{q: main question}
 How many Grover iterations should a quantum Bitcoin miner apply before measuring? 
 \end{question}
 On one hand, the more Grover iterations the quantum miner applies before measuring, the more likely another miner is able to solve the problem before the measurement step is reached. 
On the other hand, the less Grover iterations the quantum miner applies before measuring, the less advantage they have over classical algorithms.  
What makes answering Question \ref{q: main question} challenging is that the optimal choice of when to measure depends on when other miners are expected to find a marked item. 

\subsection{Our Contribution}
The key simplification we make to answer Question \ref{q: main question} is to consider the case of a single quantum miner racing to find a proof-of-work against a set of classical miners. The probability any classical miner finds a proof-of-work at time $t$ is known to approximate the exponential distribution $\lambda e^{-\lambda t}$, where $\lambda$ is the rate at which classical miners find blocks. Considering the case of a single quantum miner allows us to leverage knowledge of this distribution to understand the optimal number of Grover iterations.

We describe a procedure for quantum mining (Algorithm \ref{alg:mine}), formulate a corresponding Markov chain (Figure \ref{fig:markov}), and use this formulation to derive an expression for the expected fraction of blocks the quantum miner contributes to the blockchain (Theorem \ref{thm:success}). This analysis accounts for the quantum miner using the ``aggressive mining" techniques of \cite{Sattath}. In aggressive mining the quantum miner performs a measurement in response to a classical miner finding a proof-of-work, potentially forcing a tie in the race and yielding a reward to the quantum miner (with non-zero but low probability). 

The expression for the expected fraction of blocks the quantum miner adds to the blockchain varies with the number Grover iterations the quantum miner applies before measuring. Still, the number of Grover iterations that maximizes this expression is unclear. As a solution, we make further simplifications justified by practical considerations. 

First, we consider the regime in which the number of Grover iterations the quantum miner applies before measuring (including Grover iterations applied in parallel) is small compared to the optimal number $\frac{\pi}{4}\sqrt{D}$. We expect this to occur when the speed of the quantum miner's computer is such that the number of Grover iterations they can apply in 10 minutes, the average time for a block to be mined, is small compared to optimal. We refer to this regime as the small computational power regime as, if the quantum miner applies far less than the optimal number of Grover iterations, then they successfully mine each block with low probability and therefore make up only a small fraction of the total Bitcoin network's hash rate. It is suggested in \cite{Aggarwal} that a quantum miner having small hash rate compared to network is the most likely scenario even under optimistic assumptions of improvements in quantum technologies.   In the small computational power regime, the probability the quantum miner's measurement yields a marked item scales approximately quadratically with the number of Grover iterations applied beforehand. 

Using this approximation we derive a closed-form expression for the expected fraction of blocks the quantum miner successfully mines (Eq. \ref{eq:probility of success}).
We also show how this expression can be used to determine the effective hash rate (Eq.\ \ref{eq:eff hash}) and efficiency, or cost per block, (Eq.\ \ref{eq:Q efficiency}) of the quantum miner. By comparing the efficiency of the quantum miner with that of a classical miner, we derive conditions for quantum mining to be advantageous over classical miners (Eq.\ \ref{eq:1conditions}).

Our second simplification is to assume that only peaceful (i.e.\ non-aggressive) mining occurs, meaning, if a classical miner finds a block first, the quantum miner looses the race and does not mine a block. This assumption is accurate if the time to preform a measurement prohibits the quantum miner from quickly responding to a classical miner adding a block, if the quantum miner is not highly connected in the Bitcoin network,  or if network protocols are adjusted to prevent aggressive mining (e.g.\  the countermeasures suggested in \cite{Sattath}). We show, for a peaceful quantum miner with small hash rate compared to the network, that the expected fraction of blocks the quantum miner adds to the blockchain is maximized if they apply Grover iterations for
\begin{equation}\label{eq:main answer}
     \left(W\left(\frac{-2}{e^2}\right) + 2\right)\lambda^{-1} \approx 16 \text{ minutes}
\end{equation}
before measuring, where $W$ is the Lambert $W$-function (Theorem \ref{thm:max}).  Eq.\ \ref{eq:main answer} definitively answers Question \ref{q: main question} in the setting we consider. Remarkably, the quantum miner should measure \textit{after} $\lambda^{-1}_0 := 10$ minutes, which is the average time for the network to mine a block. 

If our optimal quantum mining procedure is used, we are able to derive a simple, easy to compute approximation for the expected fraction of blocks the quantum miner contributes to the blockchain (Eq. \ref{eq:x/a}). This approximation yields that blocks are cheaper to mine with a quantum computer than with a classical computer if
\begin{equation}\label{eq:introcond}
    Q_\$  < C_\$rb
\end{equation}
where $Q_\$$ and  $C_\$$ are the cost of Grover iterations and classical hashes respectively, $r$ is the speed of the quantum computer in Grover iterations per second, and $ b:= \lambda^{-1}_02.59\dots$. Recall, $\lambda_0^{-1} := 10$ minutes is the Bitcoin network's inter-block time.

We refer to the factor $rb$ as the quantum advantage as it is the  ratio of the average number of hashes a classical miner requires to mine a block to the average number of Grover iterations the quantum miner requires. Our results show that quantum advantage is independent of problem difficulty $D$. The intuition behind this independence is that the optimal quantum miner applies Grover iterations for a constant amount of time (16 minutes). Thus, the quantum miner does not apply more Grover iterations if difficulty increases and  does not improve their quantum advantage accordingly. In contrast, we find quantum advantage is approximately linearly related to the speed of the quantum miner. This makes sense as increases in speed allow more Grover iterations to be performed in 16 minutes. Finally we note that quantum advantage is linearly related to the average inter-block time of the network $\lambda^{-1}_0$, and so blockchains with larger inter-block times are more amenable to quantum mining. 

Our conditions for advantegous mining lay the foundation for understanding when the first quantum miners could arise. Doing so is important because quantum mining brings with it security risks and other disruptions to the Bitcoin protocol. The primary security risk associated with quantum mining is stale blocks arising from aggressive mining, a vulnerability  first identified in \cite{Sattath}. Whereas our condition is both sufficient and necessary for advantageous peaceful mining,  they are only sufficient for advantageous aggressive mining. Among the other disruptions quantum mining brings is that the assumptions that mining is ``progress free" and that inter-block times obey an exponential distribution no longer hold in the presence of quantum miners, creating uncertainties in Bitcoin's security. 

As a demonstration of the applicability of our results, we analyze the mining performance of an example quantum computer near the current difficulty $D = 10^{20}$ \cite{difficulty}. The quantum computer we consider has gate speed given in \cite{Aggarwal} equal to that of current quantum computers, and requires zero error correction. By Eq. \ref{eq:introcond}, we find that this quantum computer would require an efficiency better than 10 $\mu$J to  outperform current classical miners in mining efficiency. Furthermore, if this quantum computer used zero parallelization of Grover search, it would have an effective hash rate of $78$ MH/s, much smaller than the current network hash rate of $1.7 \times 10^11$ MH/s.

\subsection{Comparison with Related Works}

In \cite{Lee}, Lee et al.  consider Question \ref{q: main question} in the setting in which all miners are quantum. They do this by formulating a game called a quantum race. Lee et al. characterize the Nash equilibrium of quantum races that yield no payout in a tie, and therefore correspond with peaceful mining. This equilibrium is then be used to find an approximate Nash equilibrium for the general case in which a tie could payout. 
 
Aggarwal et al. \cite{Aggarwal} give a detailed analysis of the security risks quantum poses to Bitcoin, including mining attacks and attacks on Bitcoins digital signature scheme. They estimate the resources a quantum computer would need to attack Bitcoin by dominating the computing power of the network. To do so, they calculate the rate at which a quantum miner mines blocks as one divided by the time for a quantum miner to complete Grover's algorithm. Their method is useful for determining if a  quantum miner has large computing power compared to the network and Aggarwal et al. infer that the quantum computers are not powerful enough dominate the Bitcoin network at current speeds. However, in the small computational power regime we analyze, their method is less effective. We find that a quantum miner with small computational power would have a substantially lower effective hash rate than suggested by \cite{Aggarwal}. The discrepancy arises in part because the quantum miner does not have enough time to complete the entirety of Grover's algorithm if their computational power is small compared to the Bitcoin network. Therefore, the quantum miner can not reap the full $\sqrt{D}$ benefits of the quantum speed-up.

 A short analysis of the feasibility of quantum Bitcoin mining is given in \cite{Tessler2017}  as part of a larger discussion on the effects of quantum on Bitcoin. In \cite{Cojocaru2020} Cojocaru et al. show that many of the same security assumptions that rely on the dynamics of classical mining hold in the presence of quantum adversaries. They proceed by formulating generalizations of unstructured search which arise in quantum mining, and lower bounding the queries a quantum adversary would need to solve these problems. In particular, the lower bounds they develop rely on the idea that a quantum query is worth at most $O(\sqrt{D})$ classical queries. 
 
 Contrary to our focus on mining,  a variety of works have investigated the security risks Shor's quantum algorithm for factoring poses to the Bitcoin signature scheme and possible ways to address these risks \cite{Aggarwal, Stewart2018, Fedorov2018,Semmouni2019, Ilie2020, Gao2018}. Another researched  intersection of quantum and Bitcoin is quantum schemes for blockchains in which some portion of the a blockchain protocol is made quantum \cite{Jogenfors2019,Ikeda2017, Gao2020, Rajan2019}. For a more technical description of \cite{Aggarwal,Lee, Sattath} and their relation Bitcoin security, to see Section \ref{sec:Bitcoin Security}.

\subsection{Conventions}
We define $[N] := \{0,1,\dots, N-1\}$ for $N \in \mathbb N$ and define the binary alphabet to be $\mathbb B:= \{0,1\}$. All strings in our paper are over the binary alphabet. Variables representing strings are always written using typewriter typeface. SHA-256$: \mathbb{B}^* \to \mathbb B^{256}$ is the cryptographic hash function published by the National Institute of Standards and Technology in the SHA-2 standard \cite{NIST}, and we define 
\begin{equation}
\text{SHA-256}^2:= \text{SHA-256} \circ \text{SHA-256}.
\end{equation}
We say $\mathtt s_1$ is the hash of $\mathtt s_2$ if $\mathtt s_1 = \text{SHA-256}^2(\mathtt s_2)$ as we always use $\text{SHA-256}^2$ for our hashing function.
\section{Background}
\label{sec:background}
\subsection{Bitcoin Basics}
In this subsection we give an overview of the Bitcoin protocol and blockchain.
The basic construct in a blockchain is a hash pointer, which is a tuple containing a pointer to some data and the hash of that data. A hash pointer is evidence of tamper-free data as changing data almost certainly changes the hash of that data. A blockchain, which we describe in more detail below, is a linked list of these pointers. 

Define a block header to be a string encoding a tuple of four strings $\mathtt H_i = (\mathtt N_i,\mathtt P_i,\mathtt R_i,\mathtt T_i) \in \mathbb B^*$. Note that this description of a header is a slight simplification of what is actually used in the protocol as we leave out header components which are not important to proof-of-work. The elements of this tuple are referred to respectively as the nonce, previous block hash,  transaction Merkle root, and timestamp. A block is a hash pointer $B_i = (\mathtt H_i, L_i)$ where $L_i$ is a list of transactions that are hashed into $\mathtt R_i$ via a Merkle tree, which we describe later. A blockchain is a linked list of blocks $(B_i)_{i \in [M]}$ such that the hash of $\mathtt H_i$ is $\mathtt P_{i+1}$, for all $i \leq M-1$. 

In the Bitcoin protocol, there are multiple players (nodes) in a network who each store a copy of the blockchain. A block $B_{M}$ can be added to the end of a blockchain $(B_i)_{i \in [M]}$ if $\text{SHA-256}^2(\mathtt H_{M}) < \tau$ where $\tau$ is a threshold dictated by the protocol. As SHA-256 is assumed to be a noninvertible function, finding such a header is challenging, and becomes more difficult the lower the value $\tau$ takes.

To find a valid header, a miner searches over all headers  $\mathtt H = (\mathtt N,\mathtt P,\mathtt R,\mathtt T)$ such that $\mathtt R$ is the Merkle root of the list of transactions $L$ the miner wants to add to the blockchain, and $\mathtt T$ is a valid timestamp, i.e., a timestamp which is greater than the median timestamp of the last 11 blocks and less than two hours ahead of the median time registered on nodes connected to the miner. In mining the string $\mathtt P$ is fixed to be the hash of the previous block header and the nonce $\mathtt N \in \mathbb B^{32}$ can be varied with no restrictions. In fact, the purpose of the nonce is to give the miner a register that can be freely changed while mining.  The goal is to find some $\mathtt H$ such that  SHA-256$(\mathtt H) < \tau$, so that $B = (\mathtt H, \mathtt L )$ can be added to the blockchain. The significant computational resources solving this search problem requires create the security of the Bitcoin protocol. 

Each block contains a set of transactions $ L = (\mathtt t_1, \mathtt t_2, \dots, \mathtt t_n)$. Each transaction $\mathtt t_j$ is composed of a list of inputs $(I_1,I_2,\dots, I_m)$ and a list of outputs $(O_1,O_2,\dots,O_\ell)$ where each input, and each output, is a pair $(q,v)$ containing a public key $q$ and a value in Bitcoin $v$.  For each public key $q$ in an input, the transaction also contains the corresponding digital signature. Note, this signature can only be created by the holder of the private key associated with $q$. All input keys must be an output key of a previous transaction in the blockchain. It is this previous transaction that proves the holder of the public key in the input owns the Bitcoin they wish to spend.  The list of transactions also contains an additional nonce which is encoded as part of the first transaction $I_1$ (the coinbase transaction). This coinbase nonce can be varied in addition to the 32-bit block-header nonce as part of mining.

Transactions $(\mathtt t_1, \mathtt t_2, \dots, \mathtt t_n)$ are hashed into the root $\mathtt R_i$ via a Merkle tree. The Merkle tree is a hash-pointer based structure which hashes the list of transactions $L_i$ into a a single string $\mathtt R_i$ which is the root of the tree.  The leaves of the tree are the transactions $\mathtt t_j$, and each non-leaf node $\mathtt s$ is the hash of of its two children. If any transaction changes, so will the root $\mathtt R_i$ of the Merkle tree and the block header $\mathtt H_i$.  Therefore, any change in the value of the coinbase nonce propagates to the Merkle tree root along some path and involves SHA-256$^2$ calculation at every intermediate node. As a result, changes to the extra-nonce field are computationally more expensive than changes to the header nonce. 

However, as the overhead is linear in the depth of the Merkle tree, and therefore logarithmic in the number of transactions, we ignore the extra cost associated with constructing the Merkle tree for our analysis. In particular, we only count uses of SHA-256$^2$ to find the hash of headers, and do not count uses of SHA-256$^2$ required by the Merkle tree to construct headers. In this sense, the header nonce and the coinbase nonce are treated as one single nonce, whose value is determined by unstructured search.

The outputs of the coinbase transaction $I_1$ contains as an output the public key of the miner who adds a block containing the transaction to the blockchain. Through this output, if the miner adds the block to a blockchain, they receive a fee for their work.  The addition of a block to the blockchain also yields a Bitcoin reward. This reward is realized by the rule that, for the coinbase transaction, the sum of output values is equal to the sum of input values plus the reward amount.  This process provides the incentive for miners to expend computational resources in order to add blocks to the blockchain.

The threshold $\tau$, which determines mining difficulty, is chosen by the Bitcoin protocol so that a valid block is found every 10 minutes on average. To maintain this rate of new blocks,  $\tau$ is adjusted every 2016 blocks to account for the changing hash rate of the network.  In particular, if $t_0$ is time in minutes to mine a block at the previous threshold $\tau_0$, averaged over the last 2016 blocks, then the new threshold $\tau_1$ is set to $(10/t_0)\tau_0 $.

\subsection{Bitcoin Security}

The central problem all digital currencies must overcome is double spending attacks, in which currency is spent twice by the same entity. To overcome this problem, protocols can use a central authority that verifies, before every transaction, that the currency involved has not already been spent. The Bitcoin network has no such central authority and instead relies on the blockchain. The blockchain provides a public ledger which can be used to check against double spending.

The key to Bitcoins security is that altering blocks in the blockchain is computationally infeasible. If an attacker wishes to replace a block $B_j$ with $B'_j$ in a blockchain $(B_i)_{i \in [M]}$ then for all $j<k \leq M$ they must mine a new block $B'_k$.  For $j$ sufficiently less than $M$ this attack is infeasible to do, as the size of the blockchain is always increasing. Therefore blocks earlier in the blockchain are more secure.

We now discuss stale blocks, which pose a security risk to Bitcoin. Let $(B_i)_{i \in [M]}$ be a blockchain. Suppose two miners have found valid blocks $B_{M}$ and $B'_{M}$. When both miners attempt to add their block by announcing them to the Bitcoin network, a fork occurs in which the network does not agree which block should be accepted. Miners who are aware of only $B_M$ will mine on top of $B_M$. That is, they  search for a block $B_{M+1}$ with previous block hash $\mathtt P_{M+1} = \text{SHA-256}^2(\mathtt H_M)$. Miners who are aware of only $B'_M$,  mine on top of $B'_M.$ The network nodes who are aware of both $B_M$ and $B'_M$ are instructed to mine on top of whichever block they received first. This process can lead to larger forks in which $B_M$ and $B'_M$ both have blocks added on top of them. In this case, the protocol dictates that miners should always mine on top of the block in the longest blockchain. Eventually a consensus is reached by one side of the fork becoming longer, leading the blocks in the other side of the fork to become ``stale''. 

The Bitcoin network is particularly vulnerable to the 51\% attack,  which can occur when a single entity holds more that 50\% of the total network computing power. As an example, we discuss how a 51\% attack can be used to double-spend. The attacker first adds a block $B_M$ which contains a transaction for which they are an input. Once several subsequent blocks are added, the transaction is considered final, and the attacker then receives some good in exchange for their spent Bitcoin. However, as the attacker has greater than 50\% of the computing power, with high probability they can add multiple blocks on top of $B_{M-1}$, invalidating their previous transaction by creating a stale block. As a result, the attacker has received a good but not spent Bitcoin.

The 51\% attack can be  made easier by the existence of stale blocks.  This vulnerability occurs because forks divide the computational power of the network. As a result it can become easier for an adversary to overcome the networks computational power. In particular, if $p_\mathrm{stale}$ is the rate at which stale blocks occur in the network,  a 51\% attack is achievable by an attacker who has greater than $(1 - p_\mathrm{stale})/(2-p_\mathrm{stale})$ fraction of the networks computational power \cite{Sattath}.

\subsection{Grover's Search Algorithm}
In this subsection we review Grover's quantum algorithm for unstructured search and results we use in our analysis. Grover's search solves the following problem in the oracle setting.
\begin{problem}
Given an oracle $O_f$ for a function $f: \mathbb B^n \to \mathbb B$ and number of marked items $M = |f^{-1}(1)|$, find a string $\mathtt s \in \mathbb B^n$ such that $f(\mathtt s) = 1$.
\end{problem}
\noindent Grover search solves this problem using only $\left\lceil \frac{ \pi} {2} \sqrt{\frac{2^n}{M}} \right \rceil$ calls to $O_f$ which is a speedup over the $O(2^n/M)$ calls required classically.

Grover's search algorithm utilizes the operator
\begin{equation}
    G= (H^{\otimes n}) (O_f)( c^n\text{-}Z )(H^{\otimes n})
\end{equation}
where, $H$ is a Hadamard gate and $c^n\text{-}Z$ is the $n$-qubit controlled phase operator which maps $\ket{0^n} \mapsto -\ket{0^n}$ and acts trivially on all other computational basis states. We refer to $G$ as a Grover iteration.  The operator $O_f$ is a unitary operator mapping $\ket x  \mapsto  -1^{f(x)}\ket{x}$. Let $\mathcal A$ be an $n$-qubit quantum register.
Grover's search algorithm is as follows.

\begin{algorithm}[H]
\label{alg:grover}
\caption{Grover's search}
\begin{algorithmic}[1]
\Require{ $O_f$, $M$, $n$}
\Ensure{With high probability: 
$\mathtt s$ such that $f(\mathtt s) = 1$
} 
\State apply $H^{\otimes n}$ to $\mathcal A$
\State apply $G^{\left\lceil \frac{ \pi} {4} \sqrt{\frac{2^n}{M}} \right \rceil }$ to $\mathcal A$
\State return measure$(\mathcal A)$
\end{algorithmic}
\end{algorithm}
\noindent The function named measure$(\mathcal A)$  performs a measurement of $\mathcal A$ in the computational basis and returns the result. An important variant of Grover's search has been developed for the case in which $M$ is unknown \cite{Boyer1998}.  This algorithm preserves the quadratic speedup over classical algorithms. For a nice geometric description of Grover's algorithm we direct the reader to \cite{Nielsen}.

Of special interest to Bitcoin mining is the performance of Grover's search when fewer than $\left\lceil \frac{ \pi} {4} \sqrt{2^n/M} \right\rceil$ Grover iterations are used. This scenario arises in Bitcoin mining since miners race to find a solution. Therefore, they may not have time to apply the full search algorithm. We make use of the following theorem.

\begin{theorem}[Eq.\ (26) in \cite{Gingrich}] \label{thm:prob}
For a size $2^n$ search problem, measuring after $K$ Grover iterations yields a marked item with probability 
\begin{equation}
 \sin^2(2(K + 1/2)\theta) 
\end{equation}
where $\theta = \sin^{-1}(1/\sqrt{2^n/M})$.
\end{theorem}

The work of \cite{Gingrich} also shows that Grover search can be parallelized across $m$ quantum computers by having each quantum computer search over the entire $2^n$ solutions, but use only $O\left(\sqrt{\frac{2^n}{Mm}}\right)$ Grover iterations. This result shows that parallelizing Grover's search yields a factor of $\sqrt{m}$ improvement in circuit depth, whereas classically parallelizing search gives a factor of $m$ improvement.
 
\subsection{Quantum Attacks} \label{sec:Bitcoin Security}

In \cite{Aggarwal}, Aggarwal et al.  analyze the effect of quantum computation on the security of Bitcoin. As part of this analysis, they calculate the number of qubits and gate speed a quantum computer would need to dominate the Bitcoin network's computing power and, as a result, be capable of a 51\% attack. Using this analysis, they predict that ``it will be some time before quantum computers out compete classical machines for this task and when they do, a single quantum computer will not have majority hashing power." The analysis accounts for both parallelization of Grover's search and the overhead from quantum error correction.

In \cite{Sattath}, Sattath takes a different approach to the determining the implications of quantum computing to Bitcoin security. He shows that the properties of quantum Bitcoin mining lead to security risks even if no single miner dominates the networks computing power. This reduction in security comes from stale blocks. If a quantum miner receives a block midway through performing Grover's search algorithm, they can perform a measurement of their quantum register, and with probability higher than random guessing, mine a block. This procedure is called aggressive mining, whereas peaceful mining refers to the procedure in which the quantum miner ends their search if they learn of another block. If many quantum miners follow an aggressive procedure, then with higher probability than classically, two blocks are found at nearly the same time by different miners, creating many stale blocks. This result is important because it shows quantum mining poses security risks  simply because the process of quantum mining is different than classical mining, not because quantum mining is more effective.

Sattath conjectures that if quantum miners must commit to the number of Grover iterations they will apply before measuring, then the stale block rates are minimal. He proposes a change in the Bitcoin protocol that would force quantum miners to decide before they begin mining a block how many Grover iterations they will apply.  This change in protocol is incompatible with proposed fixes for insecurities arising from selfish mining techniques \cite{Eyal}.

Determining the rate of stale blocks with Sattath's proposed change is a challenging game theory question. This question is addressed in \cite{Lee}, where Lee et al. analyze a class of games they call quantum races in which multiple quantum capable players compete to solve a computational problem first. In many cases, they confirm Sattath's conjecture of low stale block rates. 

\section{Approach}

It is unlikely that in the near future a quantum computer could pose a security risk solely through dominating the Bitcoin network's computational power. However,   security risks arise in a network with many quantum miners even when no single miner dominates. Our goal is to evaluate the amenability of quantum computing to Bitcoin mining to determine conditions for a quantum computer to outperform a classical computer at mining. 

We consider a model in which there is a single quantum miner, and all other miners are classical. This model aligns with our aim because the threshold for useful quantum mining will be crossed if a single quantum computer becomes advantageous against the current network of all classical miners. We assume that the quantum miner's procedure is to repeat the process of applying $K$ Grover iterations and measuring, where $K$ is some fixed natural number.  

\subsection{Algorithm}
We first formulate the task of mining as a computational problem.  As we assume the function $\text{SHA-256}^2$ is noninvertible, we represent it as a black box function $f$ which is supplied as an oracle $O_f$.
\begin{problem}\label{pro:mining} Given  an oracle $O_f$ for a function $f: \mathbb B^* \to \mathbb B^{256}$, a threshold $\tau < 2^{256}$, a list of transactions $L$, a previous block hash $\mathtt P \in \mathbb B^{256}$, and a time range $(t_0,t_1)$,  find a nonce $\mathtt N \in \mathbb B^{32}$, a root $\mathtt R \in \mathbb B^{256}$ of a Merkle tree encoding the transactions $L$, and a timestamp $\mathtt T$ encoding a time in $(t_0,t_1)$, such that
\begin{equation} \label{eq:satisfies}
    f(\mathtt N , \mathtt P, \mathtt R, \mathtt T) < \tau.
    \end{equation}

\end{problem} 
Note that the time range $(t_0,t_1)$ corresponds to the timestamp rule described in Section \ref{sec:background}: after the median timestamp of the last 11 blocks and no more than 2 hours after the network time.
In this problem the difficulty $D = 2^n/M$ becomes $D = 2^{256}/\tau$. Note the choice $D = 2^{256}/\tau$ is different that the sometimes-used convention $D = 2^{256}/(\tau 2^{32})$. 
 We say a header is feasible if it satisfies all requirements on the output of Problem \ref{pro:mining} with the exception that the header may not satisfy Eq.\ \ref{eq:satisfies}. 
A header is valid if it is feasible and satisfies (\ref{eq:satisfies}). We also refer to a block with a valid header as a valid block. 

To successfully add a block to the blockchain the quantum miner must not only find 
a valid block, 
but do so before any other miner. As in our model all other miners are classical, the probability that a miner other than the quantum miner finds a block at time $t$ follows the exponential distribution $\lambda e^{-\lambda t}$. 

In our analysis $\lambda^{-1}$ describes the average rate at which blocks are found by any classical miner and we define $\lambda_0^{-1} := 10$ minutes, which is the average time for a block to be mined by \textit{any} miner.  
In practice $\lambda$ is easy to calculate since $\lambda = D/\mathcal P$ where $D$ is the difficulty and $\mathcal P$ is the combined hash rate  of all classical miners.  Therefore, our use of $\lambda$ as a parameter the quantum miner uses to determine their optimal protocol is justified. 

Let $\eta: ( \mathtt N, \mathtt P, \mathtt R, \mathtt T) \mapsto n \in [S]$ be an enumeration of all feasible headers where $S$ is the number of feasible headers. Let 
\begin{equation}
    g: [S] \to \mathbb B
\end{equation}
be the function that evaluates to $1$ if and only if $\eta^{-1}(n)$ satisfies Eq.\ \ref{eq:satisfies} for $n \in [S]$. An oracle $O_g$ for $g$ can be constructed which computes $\eta^{-1}$ and uses only one call to $O_f$.  Suppose $O_g$ requires $L$ ancillary qubits to implement and takes $\ell$ qubits as input. Suppose the quantum miner has a quantum computer with $m(L + \ell)$ qubits, or alternatively, $m$ quantum computers which each have $L+\ell$ qubits; these two scenarios are equivalent for our purposes.  We give the following protocol for mining.

Let $\{\mathcal A_i\}_{i \in [m]}$ be a set containing $m$ quantum registers, each of which contains $(\ell+L)$-qubits. In our procedure, the quantum miner runs a quantum search on each of these registers in parallel.  Let 
\begin{equation}
G = (H^{\otimes \ell}\otimes \mathds 1_L) O_g (c^\ell\text{-}Z\otimes \mathds 1_L)( H^{\otimes \ell}\otimes \mathds 1_L )
\end{equation}
be a Grover iteration
where $\mathds 1_L$ is the $L$-qubit identity operator. Fix a number $K \in \mathbb N$ of Grover iterations. 
We assume the quantum miner applies $K$ Grover iterations before measuring if no valid blocks are received from other miners. If a valid block is received, the quantum miner measures all of their registers if and only if mining aggressively. The value of $K$ that maximizes the quantum miner's success probability is the focus of Subsection \ref{subsec:prob suc}. Our procedure is described in Algorithm \ref{alg:mine}.

\begin{algorithm}[H]
\caption{Quantum Bitcoin Mining}
\label{alg:mine}
\begin{algorithmic}[1]
\Require{ $O_g$, $\eta^{-1}$, $\texttt{agg} \in \mathbb B$}
\Ensure{  Either ``no block found" or block header $\mathtt H$ such that $g(\mathtt H)  = 1$}
\While{true} \Comment{Loop indefinitely} 
\State{set all registers to the zero state}
\State{apply $H^{\otimes \ell}$ to each $\mathcal A_i$}
\For{$k \in [K]$}
\If{new block received} \label{ln:classical miner finds block}
\If{agg == 0} \Comment{Mine peacefully} 
\State{return ``no block found" }
\EndIf
\For{ $\mathcal A_i \in \{\mathcal A_i\}$} \Comment{Mine aggressively}
\State{ $a_i \gets \text{measure}(\mathcal A_i)$} \label{ln:agg measurement}\Comment{Measure $\mathcal A_i$}
\If{ $g(a_i) == 1$ } \Comment{Apply $O_g$} \label{ln:quantum miner also finds block}
\State{return block with header $\eta^{-1}(a_i)$} \label{ln:beta end point}
\EndIf
\EndFor
\State{return ``no block found" } \label{ln:quantum miner fails}
\EndIf
\State{apply $G$ to each $\mathcal A_i$ }
\EndFor
\For{ $\mathcal A_i \in \{\mathcal A_i\}$} \label{ln:classical miner does not find block}
\State{ $a_i \gets \text{measure}(\mathcal A_i)$} 
\If{ $g(a_i) == 1$ } \Comment{Apply $O_g$} \label{ln: quantum miner finds block}
\State{ return block with header $\eta^{-1}(a_i)$}  \label{ln:alpha end point}
\EndIf \label{ln: quantum miner does not find block}
\EndFor
\EndWhile
\end{algorithmic}
\end{algorithm}

This algorithm contains two procedures, one for aggressive mining ($\texttt{agg} = 1$), and the other for peaceful mining  ($\texttt{agg} = 0$). We consider both as peaceful mining is easier to analyze, and describes mining accurately in many scenarios, such as if the time to perform a measurement prohibits the quantum miner from effective aggressive mining. 
For both cases, we aim to determine the probability that the above procedure results in the quantum miner finding a marked header and that the resultant block gets accepted by the blockchain. We refer to this probability as the probability of success. 

\subsection{Markov Chain}
To determine probability of success, we represent the aggressive mining procedure using the state transition graph in Figure \ref{fig:markov}. If we take $\phi = 0$, then this state transition graph describes the peaceful mining procedure. We define $T := K/r$ where $K$ is the number of Grover iterations and  $r$ is the speed of the quantum miner in Grover iterations per second. The executing events in the protocols are as follows. 

\begin{enumerate}[leftmargin= 1.75cm]
    \item[$(1) \to (2)$] Some classical miner announces a valid block before $T$.
    \item[$(2) \to (6)$] The quantum miner discovers a valid block using aggressive mining before $T$.
    \item[$(2) \to (5)$] The quantum miner fails at aggressive mining, and the classically mined block is accepted to the blockchain. 
    \item[$(1) \to (3)$] No classical miner finds a block before $T$.
    \item[$(3) \to (4)$]  The quantum miner discovers a valid block at time $T$.
    \item[$(3) \to (1)$] The quantum miner's measurement at time $T$ fails to yield a valid block.
    \item[$(6) \to (7)$] The classically mined block gets accepted and the quantum miner's block becomes stale.
    \item[$(6) \to (8)$] The quantum miner's block gets accepted and the quantum miner's block becomes stale.
 \end{enumerate}

\begin{figure}
\begin{center}
\begin{tikzpicture}[main node/.style={circle,draw}, inner sep=2pt, minimum size=20pt, scale=1.25]

    \node[main node] (1) at (0,0) {\large 1};
    \node[main node] (2) at (1,1) {\large 2};
    \node[main node] (3) at (1,-1) {\large 3};
    \node[main node] (4) at (2,-2) {\large 4};
    \node[main node] (5) at (2,2) {\large 5}; 
    \node[main node] (6) at (2,0) {\large 6};
    \node[main node] (7) at (3,1) {\large 7};
    \node[main node] (8) at (3,-1) {\large 8};
\path[->,thick,shorten >= 3pt,shorten <= 3pt,>=angle 90]
(1) edge node[above left]  {$\mu$} (2) 
(1) edge[bend left=30] node[above right]  {$1- \mu$}  (3)
(3) edge[bend left=30] node[below left]  {$1- \nu$}(1)
(2) edge node[above left]  {$1- \phi$} (5) 
(2) edge node[above right]  {$\phi$} (6) 
(6) edge node[below right]  {$1 - \gamma$} (7)
(6) edge node[below left]  {$\gamma$}  (8) 
(3) edge node[below left]  {$\nu$}  (4) 
  ;
\end{tikzpicture}

\bigskip
KEY
 
\begin{enumerate}
    \item Initial state 
    \item Classical miner has found a block in time $T$ (ln \ref{ln:classical miner finds block})
    \item Classical miner did not find a block in time $T$ (ln \ref{ln:classical miner does not find block})
    \item Quantum miner finds a block (ln \ref{ln: quantum miner finds block})
    \item In response to classical miner finding a block, the quantum miner measured and failed to find a block
    (ln \ref{ln:quantum miner fails})
    \item In response to classical miner finding a block, the quantum miner measured and found a block (ln \ref{ln:quantum miner also finds block})
    \item Classical miner's block is accepted
    \item Quantum miner's block is accepted
\end{enumerate}
\begin{enumerate}
    \item[($\nu$)] Probability the quantum miner's measurement after $T$ seconds of Grover iterations yields a block to the quantum miner
    \item[($\mu$)] Probability the classical miner finds a block before $T$
    \item[($\phi$)] Quantum miner finds a block in response to classical miner finding a block first (aggressive mining)
    \item[($\gamma$)]  Fraction of miners who mine on the quantum miner's block when a fork occurs
\end{enumerate}
 States (4), (5), (7), (8) contain self loops of weight 1 (not shown)
\end{center}\caption{State transition graph for quantum Bitcoin mining. 
}\label{fig:markov}
\end{figure}

We model these events as independent, and their probabilities are in the KEY.  Quantum miner generate a valid block in two situations.
\begin{enumerate}
    \item[$(4)$] No classical miner finds a block in time $T$ (probability $1 - \mu$), and the quantum miner's subsequent measurement succeeds, which occurs with probability $\nu$.
    \item[$(8)$] Some classical miner generates a valid block before $T$ with probability $\mu$ and announces the block to the network. The quantum miner receives this valid block, and in response, performs a  measurement (aggressive mining) that succeeds with probability $\phi$.  Subsequent blocks are then added to on top of the quantum miner's block, as opposed to the classical miner's block  resulting in acceptance of the quantum miner's block(probability $\gamma$). 
\end{enumerate}

The quantum miner's probability of success  when aggressive is the sum $P := p_{14} + p_{18}$ where $p_{14}$ is the probability the resultant state is $(4)$ and $p_{18}$ is the probability the resultant state is $(8)$. As above, these are the only two states which represent success for the quantum miner. In the peaceful case, the probability of success $P$ is simply $p_{14},$ as $\phi = 0$ implies $p_{18} = 0$. 

Importantly, the state transition graph shows that if the quantum miner's measurement at $(3)$ does not return a valid block  then the resultant state is  $(1)$. To see this, first note that the quantum miner resets their quantum computer after measuring. Also exponential distributions, such as that which describes the inter-arrival times of classically found blocks, are 
memoryless. As a result, if the quantum miner's measurement at  at time $T$ fails then the probability a classical miner finds a block at some time $T' > T$ is 
\begin{equation}
    \lambda e^{-\lambda \Delta T}
\end{equation}
where $\Delta t = T - T'.$ Thus, the state following $(3)$, if the quantum miner did not find a block, is indistinguishable from the starting state $(1).$ This property is directly a result of the Bitcoin proof-of-work being progress-free \cite{Biryukov2016} for classical miners. 

Our procedure assumes a $K$ that is not changed between iterations.  However, the constant $K$ case always contains an optimal quantum mining procedure; as the state of the system directly after a failed quantum measurement is always identical, a value of $K$ that maximizes the quantum miner's future success following this measurement must also be the same as the initial choice for $K$. In short, as each time the quantum miner reaches $(1)$, the state of the system is indistinguishable, the quantum miner can always make the same choice of $K$.

\subsection{Assumptions and Approximations}

We now make rigorous the assumptions and approximations we use in our analysis of success probability.  Recall $D=2^n/M = 2^{256}/\tau$. We develop approximations for the case where $1 \ll K \ll \sqrt{D}$. This is the regime in which the quantum miner applies far fewer than the optimal number of Grover iterations before measuring, but still applies a number of Grover iterations much greater than one. We refer to this regime as the small computational power regime. 

We now motivate each part of our assumption. We take $1 \ll K$ since if $K$ is close to $1$ then the quantum miner has little advantage over classical mining. Furthermore, even a relatively slow quantum computer is still fast enough to apply many more than a single Grover iteration in the 10 minutes average time it takes for a block to be mined.  As an example, if $r$ is one Grover iteration per second then the quantum miner can apply $600$ Grover iterations in 10 minutes. 

Second, we take $K \ll \sqrt{G}$ since this aligns with the quantum computer having small computational power compared to the Bitcoin network, which is the predicted to be the most likely scenario by \cite{Aggarwal}. To see this, notice if $K$ is near $\sqrt{G}$ then quantum miner mines blocks with high probability. We expect $K \ll \sqrt{G}$ if $r/\lambda_0 \ll \sqrt D$, that is, if the number of Grover iterations  the quantum miner can apply in 10 minutes is small compared to optimal.

We take the stronger assumption that $mK \ll \sqrt{D}$ in the final step of our derivation of  conditions for advantegous quantum mining  (Eq. \ref{eq:x/a}). As $K$ is the number of sequential Grover iterations the quantum miner applies, and $mK$ is the total number of Grover iterations they apply (both sequential and in parallel) this is the regime in which the quantum miner applies far fewer Grover iterations than the optimal number of sequential Grover iterations before measuring. No other approximations in our paper require $mK \ll K$. 

We distinguish between $\lambda$, the rate at which classical miners mine blocks, and $\lambda_0$ the rate at which the entire network mines blocks. However, in the small computational power regime, the quantum miner contributes only a small portion to the rate at which blocks are mined by the network. As a result, $\lambda \approx \lambda_0$. We frequently turn to this approximation to make simplifications as $\lambda_0$ is fixed to $10$ minutes by the Bitcoin protocol. 

For all parts of our analysis, we only consider the time cost to implement Grover iterations. We expect this time cost to dominate the time used by other operations such as measuring quantum registers, resetting quantum registers, and inverting $\eta$. This assumption is accurate in the $1 \gg K$ limit we consider.

\section{Results}

\label{sec:results}

\subsection{Probability of Success}
\label{subsec:prob suc}

In this subsection we derive the expression for the quantum miner's probability of successfully adding a block to the blockchain, or in other words, the expected fractions of blocks they add to the blockchain. Next, we develop a closed-form approximation for this success probability valid in the small computational power regime. Finally, we determine the value of $K$ that optimizes this approximation in the case that the quantum miner is peaceful. 

 \subsubsection{Expressions for Transition Probabilities}
 
 We begin our determination of success probability by considering the expression for $\nu$. Using Theorem \ref{thm:prob}, the probability a measurement of a single register $\mathcal A_i$ at time $T$ yields a marked header is
 \begin{equation} \label{eq:q success m=1}
\sin^2\left(2(K + 1/2)\theta \right)
\end{equation}
where $\theta = \sin^{-1}\left(1/\sqrt D\right)$.  
As $m$ registers are measured at time $T$, the total probability the measurements yields a valid block is
\begin{equation} \label{eq:q success}
    \nu = 1- \left(1- \sin^2\left(2(K + 1/2)\theta \right)\right)^m.
\end{equation}

 Next,  as $\sqrt{D} \gg 1$, applying the small angle approximation yields $\theta \approx 1/\sqrt D$. Using $K \ll \sqrt{D}$ to apply the small angle approximation once more we get
 \begin{equation}
    \sin^2\left(2(K + 1/2)\theta \right) \approx  \frac{4K^2}{D}.
\end{equation}
Note, we have also used that $K \gg 1$ implies $K + 1/2 \approx K$.
The probability the measurements at $T$ yields a valid block is then approximately 
\begin{equation}
   \nu \approx 1 - \left(1 - \frac{4K^2}{D}  \right)^m.
\end{equation}

As $\frac{4K^2}{D}$ is close to zero in the small computational power regime, we approximate $\left(1 - \frac{4K^2}{D}\right)^m$ by its binomial series and take only the first two terms to get
\begin{equation}
 \nu \approx 1 - \left(1 - \frac{4K^2}{D}  \right)^m \approx \frac{4mK^2}{D}.
\end{equation}
The approximation used here is similar a common approximation used for classical mining where it is often assumed that parallelizing over $m$ processors boosts the probability of success by a factor of $m$. In fact, this is only correct if a classical miner checks a \textit{random} header each hash, and therefore may check the same header twice.
We define our approximation $\tilde \nu$ as
\begin{equation}
    \tilde \nu := \frac{4mK^2}{D}.
\end{equation}
 
 Now we discuss the expression for the probability $\mu$ that the classical miner finds a block before time $T$. The time for a block to be mined mined by a classical miner follows an exponential distribution $\lambda e^{-\lambda t}$ where $\lambda^{-1}$ is the average time for a block to be mined classically. Therefore, the probability any classical miner finds a block before time $T$ is
 \begin{align}
     \mu &= \int_0^{T} \lambda e^{-\lambda t}dt\\
     &= 1 - e^{-\lambda T}.
 \end{align}
 
 Finally, we find an expression  for the probability $\phi$ that if a classical miner finds a block first, the quantum miner also finds a block via aggressive mining. If the quantum miner is following the peaceful protocol, then $\phi = 0$, so here we address the aggressive case. Recall from the analysis of $\nu$ that the quantum miner's measurement after applying $k$ Grover iterations yields a marked header with probability 
 \begin{equation}
     q(k) = 1- \left(1- \sin^2\left(2(k + 1/2)\phi \right)\right)^m
 \end{equation}
 or, using the substitution $t := k/r$,
  \begin{equation}
     q(t) = 1- \left(1- \sin^2\left(2(rt + 1/2)\phi \right)\right)^m.
 \end{equation}

 The number of Grover iterations $k$ the quantum miner has applied when they measure at line \ref{ln:agg measurement} is dependent on the time $t$ at which the classical miner finds a block. In particular, conditioned on a classical miner finding a block at time $t$, the quantum miner finds an additional block with probability $q(t)$. The time $t$ is drawn from the probability distribution $\lambda e^{-\lambda t}$. However, as $\phi$ is the probability of the transition $(2) \to (6)$ it must be that $t < T$ as this is required for the state to be $(2)$. Therefore, the probability distribution of times must be normalized by a factor $A$ to account for this restriction:
 \begin{align}
     A\int_0^{T} \lambda e^{-\lambda t}dt = 1\\
     \implies A  = \frac{1}{1 - e^{-\lambda T}}.
 \end{align}
The probability the classical miner finds a block at time $t$, conditioned on $t < T$, is 
\begin{equation}
    \frac{1}{1 - e^{-\lambda T}}\lambda e^{-\lambda t}. 
\end{equation}
The probability the quantum miner finds an additional block is then
 \begin{equation}\label{eq:phi}
     \phi = \int_0^{T}    \frac{\lambda}{1 - e^{-\lambda T}} e^{-\lambda t} q(t)dt 
 \end{equation}
 by the law of total probability. 
 
 Using the same approximations used in the derivation of $\tilde \nu$,
 \begin{equation} \label{eq:q(t)}
     q(t) \approx \frac{4m(rt)^2}{D}.
 \end{equation}
 Substituting Eq.\ \ref{eq:q(t)} into Eq.\ \ref{eq:phi} yields our approximation $\tilde \phi$ for $\phi$:
\begin{align}
     &{\tilde \phi} := \int_0^{T}    \frac{\lambda}{1 - e^{-\lambda T}} e^{-\lambda t} \frac{4m(rt)^2}{D} dt\\
     &= \frac{4mr^2\lambda}{D(1 - e^{-\lambda T})}\int_0^{T}    e^{-\lambda t} t^2 dt
    \\
    &= \frac{4mr^2}{D(1 - e^{-\lambda T})}\frac{e^{-T\lambda}(-T\lambda(T\lambda +2) -2) + 2}{\lambda^2}.
 \end{align}

\subsubsection{Total Probability of Success}

Now that we have expressions for the transition probabilities in Figure \ref{fig:markov}, as well as approximations for these expressions in the $\sqrt D \gg K$ and $K \gg 1$ limit, we  analyze the total probability that the quantum miner successfully adds a block to the blockchain.  To start, we prove the following lemma. 

\begin{lemma}
Starting at state (1), after an infinite number of steps the probability the state is (4) is 
\begin{equation}
    p_{14} = \frac{\nu(1-\mu)}{1- (1-\mu)(1-\nu)}.
\end{equation}
\end{lemma}
\begin{proof}
As the only sequences of states which result in a final state of (4) are (1)(3)(4), (1)(3)(1)(3)(4), (1)(3)(1)(3)(1)(3)(4) and so on, we see that the state (4) can only be reached after an even number of steps. Let $w(n)$ be the probability the state (4) is reached after $n$ steps. Then, the total probability that the final state is (4) is  given by 
\begin{equation}\label{eq:p14}
    p_{14} = w(2) + w(4) + w(6) + \dots.
\end{equation}
The value of $w(n)$ for $n$ even is 
\begin{equation} \label{eq:q(n)}
    w(n) = \nu(1- \mu)\left((1 - \mu)(1-\nu)\right)^{n/2}
\end{equation}
as the only way to arrive at (4) after $n$ steps is to follow the sequence 
\begin{equation}
   \underbrace{ (1)(3)(1)(3)\dots(1)(3)}_{\text{$n/2$ times}}(4).
\end{equation}
Substituting Eq.\ \ref{eq:q(n)} into Eq.\ \ref{eq:p14}  yields a geometric series:
\begin{widetext}
\begin{align}
     p_{14} &= \nu(1 - \mu)(1 - \mu)(1-\nu) + \nu(1 - \mu)\left((1 - \mu)(1-\nu)\right)^2 + \nu(1 - \mu)\left((1 - \mu)(1-\nu)\right)^3 + \dots \\
     &= \nu(1 - \mu)\left( (1 - \mu)(1-\nu) + \left((1 - \mu)(1-\nu)\right)^2 + \left((1 - \mu)(1-\nu)\right)^3 + \dots\right)\\
     &= \frac{\nu(1 - \mu)}{1 - (1 - \mu)(1-\nu)}.
\end{align}
\end{widetext}
\end{proof}

We now evaluate $p_{14}$ using our expressions and approximations of transition probabilities. Substituting in the expression for $\mu$ and using $\tilde \nu$ as an approximation of $\nu$ gives that the approximation 
\begin{align}
   \tilde p_{14} :&=\frac{\tilde \nu(1 - \mu)}{1 - (1 - \mu)(1-\tilde \nu)} \\
   &=\frac{\frac{4mK^2}{D}e^{-\lambda K/r}}{1-e^{-\lambda K/r}(1- \frac{4mK^2}{D})}
\end{align}
for the probability that the final state is (4). 

We now derive the total probability that the quantum miner successfully mines a block. 
\begin{theorem}\label{thm:success}
The probability $P := p_{14} + p_{18}$ that the quantum miner successfully adds a block to the blockchain is

\begin{equation}
    P= \frac{\nu(1-\mu)}{1- (1-\mu)(1-\nu)} + \gamma \phi \left(1 -  \frac{\nu(1 - \mu)}{1 - (1 - \mu)(1-\nu)}\right).
\end{equation}
\end{theorem}
\begin{proof}
Examination of the state transition graph reveals that the state   (2) 
 is visited 
if and only if the final state  (4) is not. Therefore if $p_{12}$ is the probability that (2) is reached at any number of steps from (1) then 
\begin{equation}
    p_{12} = 1 - p_{14}.
\end{equation}
The state transition graph shows that the probability $p_{18}$ that the final state is (8) is given by
\begin{align}
   p_{18} &=  p_{12}\phi\gamma   \\
    &= \left(1 -  p_{14}\right)\phi\gamma  ,\\
    p_{18} &=  \left(1 -  \frac{\nu(1 - \mu)}{1 - (1 - \mu)(1-\nu)}\right)\phi\gamma . \label{eq:p18}
\end{align}
Substituting this expression for $p_{18}$ into $P := p_{14} + p_{18}$ yields the statement of the theorem. 
\end{proof}

We now derive a closed-form expression for the quantum miner's success probability under the approximations $\phi \approx \tilde \phi$ and $\nu \approx \tilde \nu$. 
First, our approximation $\tilde p_{18}$ for $p_{18}$ is found by
evaluating $\mu$ and substituting $\tilde \nu$ and $\tilde{\phi}$ for $\nu$ and $\phi$ in Eq. \ref{eq:p18} yielding
\begin{widetext}
\begin{align}
\tilde p_{18} :&= \left(1 -  \frac{\tilde \nu(1 - \mu)}{1 - (1 - \mu)(1-\tilde \nu)}\right)\tilde \phi\gamma  \\
 &= \left(1 -  \frac{\frac{4mK^2}{D}e^{-\lambda K/r}}{1-e^{-\lambda K/r}(1- \frac{4mK^2}{D})}\right)\frac{4mr^2\gamma}{D\lambda^2(1 - e^{-\frac{\lambda K}{r}})}\left(e^{-\frac{K\lambda}{r}}\left(-\frac{K\lambda}{r}\left(\frac{K\lambda}{r} +2\right) -2\right) + 2\right).
\end{align}
Therefore, the total probability $P := p_{14} + p_{18}$ of the quantum miner's success is approximated by

\begin{align} \label{eq:probility of success}
    \tilde P :&= \tilde p_{14} + \tilde p_{18}\\
 &=  \frac{\frac{4mK^2}{D}e^{-\lambda K/r}}{1-e^{-\lambda K/r}(1- \frac{4mK^2}{D})} +\left(1 -  \frac{\frac{4mK^2}{D}e^{-\lambda K/r}}{1-e^{-\lambda K/r}(1- \frac{4mK^2}{D})}\right)\frac{4mr^2\gamma}{D\lambda^2(1 - e^{-\frac{\lambda K}{r}})}\left(e^{-\frac{K\lambda}{r}}\left(-\frac{K\lambda}{r}\left(\frac{K\lambda}{r} +2\right) -2\right) + 2\right).
\end{align}

\end{widetext}

\subsubsection{Optimal Number of Grover Iterations}

Now, we find the value of $K$ which maximizes $\tilde p_{14}$. 
This value describes the optimal $K$ for quantum mining when $\gamma$ is small, or when the quantum miner uses only peaceful mining. 
\begin{theorem} \label{thm:max}
The probability 
\begin{equation}
   \tilde p_{14} = \frac{\frac{4mK^2}{D}e^{-\lambda K/r}}{1-e^{-\lambda K/r}(1- \frac{4mK^2}{D})}. 
\end{equation}
attains its maximum value at $K = \dfrac{y_0 r}{\lambda}$ where 
\begin{align}
    y_0 :&= W\left(\frac{-2}{e^2}\right) + 2
   \\ &\approx  1.59362426
\end{align}
where $W$ is the Lambert $W$-function, i.e., the inverse function of $f(z) = ze^z$. 
\end{theorem}
\begin{proof}
We begin with the variable substitutions 
\begin{align}
    x &:= \frac{4mr^2}{\lambda^2 D}\\
    y &:= \lambda K/r. \label{eq:y}
\end{align}
These substitutions yield
\begin{equation}
    \tilde p_{14} =\frac{xy^2e^{-y}}{1 - e^{-y}(1 - xy^2)}. 
\end{equation}

The variable $x$ is a measure of the quantum computers power in relation to problem difficulty. As $m,r,\lambda, D >0$, the value of $x$ is always positive. The variable $y$ determines the time at which the quantum miner should measure, as $T = y/\lambda$. 
Simplifying further,
\begin{align}
\tilde p_{14} &= 
  \frac{xy^2e^{-y}}{1 - e^{-y}(1 - xy^2)}  \\&=\frac{xy^2}{e^y + xy^2 - 1}.
\end{align}
Figure \ref{fig:plot} (from Wolfram) shows the plot of $\tilde p_{14}$ if $x$ is in the range $[0,1]$ and $y$ is the range $[1,10]$.

\begin{figure}
    \centering
    \includegraphics[width=0.4\textwidth]{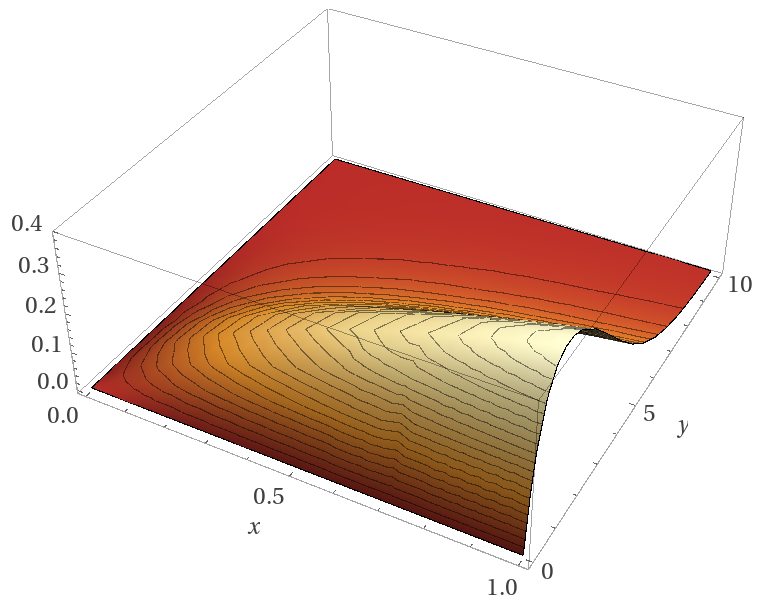}
    \caption{Probability Plot}
    \label{fig:plot}
\end{figure}

To find the maximum, we take the derivative with respect to $y$ and set this derivative equal to zero:
\begin{align} \
    0 &= \frac{\partial}{\partial y} \frac{xy^2}{e^y + xy^2 - 1} 
    \\ &= \frac{-(x (2 + e^y (-2 + y)) y)}{(-1 + e^y + x y^2)^2}. \label{eq:for alpha}
\end{align}
 Solving Eq.\ \ref{eq:for alpha} with computer algebra software  (Wolfram Alpha \cite{WolframAlpha}), we find only one real root at $y = y_0$ where $y_0$ is given in the theorem statement.  From Figure \ref{fig:plot}, we see that $y_0$ corresponds to a maximum of $\tilde p_{14}$, not a minimum. Interestingly, the maximum at $y = y_0$  is independent of $x$, and therefore the optimal mining procedure is independent of both problem difficulty and the computational power of the miner.
 
Solving Eq.\ \ref{eq:y} for $K$ gives that the maximum occurs at
\begin{align}
 K = \frac{y_0 r}{\lambda}.
\end{align}
\end{proof}
As $r$ is the quantum computer's clock rate in Grover iterations per second, and $\lambda^{-1} \approx \lambda_0^{-1} = 10$ minutes in the small computational power regime, this root implies that the quantum miner should apply Grover iterations for ${y_0 \cdot 10 \approx 16}$ minutes before measuring. One might wonder what the probability is of classical miners failing to find a block for the $16$ minutes that the quantum miner takes to reach the measurement step. This probability is given by
\begin{align}
    1 - \int_0^{y_0/\lambda_0}\lambda_0 e^{-\lambda_0t}dt &= e^{-y_0}\\
    &\approx .203.
\end{align}
In other words, if the quantum miner plans to measure at $16$ minutes, then there is approximately a 20\% chance a classical miner does not find a block before they make this measurement.

\subsection{Performance Measures}

\subsubsection{Effective Hash Rate}
The probability the quantum miner succeeds $P$ easily translates into an effective hash rate. We define the quantum miner's effective hash rate to be the hash rate of a classical computer which successfully mines the same proportion of blocks.  As $D$ is the average number of classical hashes required to mine one block and $\lambda_0 := 1/10$ minutes is the average time for a block to be mined, the total hash rate of the network is given by $D\lambda_0$. The effective hash rate of the quantum miner is then simply
\begin{equation}\label{eq:eff hash}
   PD\lambda_0
\end{equation}
as $P$ is the expected fraction of blocks the quantum miner adds to the blockchain. 
Note that this definition for effective hash rate is different from the one in \cite{Aggarwal}.

\subsubsection{Efficiency}

Now we analyze the efficiency advantage a quantum computer provides in mining Bitcoin. 
On average the quantum miner applies $mr/\lambda_0$ Grover iterations for each block added to the blockchain (by any miner) since the quantum miner is continually applying $m$ Grover iterations in parallel at a rate $r$ and each block takes $\lambda_0$ average time to be mined. The quantum miner adds $P$ portion of these blocks. Thus, the expected number of Grover iterations applied for each block the \textit{quantum miner} adds to the blockchain is
\begin{equation} \label{eq:Qcost}
    \frac{mr}{\lambda_0 P}. 
\end{equation}
 Let $Q_\$$ be the cost, in dollars, to implement one Grover iteration. This cost, like the cost of classical mining, is primarily derived from energy expenditures. The expected cost per block for the quantum miner is then
\begin{equation} \label{eq:Q efficiency}
    \frac{Q_\$mr}{\lambda_0 P}.
\end{equation}

Now compare this mining efficiency with that of a classical miner. Let $C_\$$ be the cost in dollars to implement one hash on a classical computer. The probability that a single hash yields a marked header is $1/D$, so we expect the classical miner to use $D$ hashes per block they mine. Therefore, the expected efficiency for the classical miner is 
\begin{equation}
   C_\$D
\end{equation}
and the condition for the quantum miner to be more efficient at mining than a classical miner is then
\begin{equation}\label{eq:1conditions}
    \frac{Q_\$mr}{\lambda_0 P} < C_\$D.
\end{equation}

Now, we once again consider the peaceful mining case where $\phi = 0$ and $P = p_{14}$.
In order to understand the scaling in this case more clearly, we approximate $p_{14}$ with
\begin{equation}
    \tilde p_{14} = \frac{xy^2}{e^y + xy^2 - 1}.
\end{equation}

If the quantum miner uses the optimal procedure, that is they measure after approximately 16 minutes, then $y = y_0$. Using $y = y_0$ in the expression for $\tilde p_{14}$ and reducing yields, 
\begin{align}
    \tilde p_{14} = \frac{x}{a + x}
\end{align}
where 
\begin{align}
    a :&= \frac{e^{y_0} - 1}{y_0^2}\\
     &\approx 1.544. 
\end{align}

We now extend our $K \ll \sqrt{D}$ assumption to  $mK \ll \sqrt D$. As, for the optimal quantum miner, $mK = my_0r/\lambda$, this extended assumption implies $mr/\lambda \ll \sqrt D$.  Taking the first term from the Taylor series expansion of $x/(a+x)$ about $x=0$ gives
\begin{align}\label{eq:x/a}
    \tilde p_{14} \approx x/a
     = \frac{4mr^2}{a\lambda^2 D},
\end{align}
and approximation that becomes accurate in the $mr/\lambda \ll \sqrt D$ limit. 
Under this same approximation, Eq.\ \ref{eq:1conditions} with $P \approx \tilde p_{14}$ and $\lambda$ approximated by $\lambda_0$ becomes 
\begin{equation} \label{eq:effecient}
  \frac{a\lambda_0 Q_\$}{4r} <  C_\$ 
 \end{equation}
Plugging in the value of $a$, this expression becomes
\begin{equation}
    Q_\$  < C_\$\frac{r}{\lambda_0}2.59\dots
\end{equation}
where the units of $r$ are Grover iterations per second.

The expectation is that as quantum computers improve the cost $Q_\$$ of a Grover iteration decreases, and the  rate of Grover iterations $r$ increases over time, potentially making quantum mining viable in future. By projecting values of $r, Q_\$,$ and $C_\$$ into the future, estimates can be made for when quantum mining will become advantegous. While we do not compute these projections ourselves, in the next section we evaluate the energy efficiency a quantum computer would require for advantageous mining using the current cost of classical hashes $C_\$$ and optimistic assumptions for $r$.

\subsection{Example Application}

In this subsection we demonstrate an application of our results by calculating numerical estimates for a quantum miner's performance. These calculation demonstrate how to use our results to evaluate the feasibility of a given quantum computer for Bitcoin mining. We give estimates for both  effective hash rate and energy efficiency required for advantegous mining. 

The quantum computer we consider is described in \cite{Aggarwal}. First, the computer has a gate speed of $66.7$ MHz, which is the speed achievable on current devices. Aggarwal et al. also show that a single Grover iteration (for the Bitcoin search problem) would take a circuit of depth $297784$ to perform if we assume no overhead from error correction.  We make this assumption for simplicity as the error correction overhead has a non-trivial relationship with the number of sequential Grover iterations used. For this quantum computer, 
\begin{align}
    r &= \frac{ 66.7 \times 10^{6}}{297784} \\ 
    &= 224 \text{ Grover iterations per second}.
\end{align}

Recall, the effective hash rate of the quantum miner is
 $ P\lambda D$ (Eq.\ \ref{eq:eff hash}). In the peaceful mining case, $P = p_{14}$ which can be approximated by $\tilde p_{14} = x/(a+x)$.
Assuming $m = 1$ (no parallelism), and near the current difficulty at $D = 10^{20}$ \cite{difficulty}, we find 
\begin{align}
    x &= \frac{4(224.0 )^2\cdot (10 \cdot 60)^2  }{ 10^{20}}\\
    &= 7.2 \times 10^{-10}
\end{align}
and 
\begin{align}
    \tilde p_{14} &= x/(a+x) \\ 
    &= 4.7\times 10^{-10}.
\end{align}
This means the quantum computer would comprise only a small fraction  of the  mining power of the Bitcoin network. 
Finally, we calculate the effective hash rate to be
\begin{equation}
    \tilde p_{14}\lambda D \approx 78\text{ MH/s} 
\end{equation}
Compared with the total network hash rate of $D\lambda = 1.7 \times 10^{11}$ MH/s, the quantum miner's effective hash rate of $78$ MH/s is minuscule. 

\paragraph{Efficiency Requirement}
Next we turn to the efficiency the quantum computer would need to outperform a classical computer at Bitcoin mining. Recall the condition for this outperformance is given by Eq.\ \ref{eq:1conditions}. Plugging into this equation we find the condition
\begin{equation} \label{eq:condition}
    Q_{\$} <  C_\$3.49 \times 10^5
\end{equation}
for advantageous quantum mining. If we instead use Eq.\ \ref{eq:effecient} which employs an additional approximation, then we get same result, up to three significant figures. 

Eq.\ \ref{eq:condition} states that for a quantum miner to be more efficient than a classical miner,  the energy cost of a Grover iteration must be no more than $3.49 \times 10^5$ times more expensive than the cost of a classical hash. Current specialized classical mining equipment has an energy efficiency on the order of $10^{-10}$ Jules per hash \cite{miners}. Thus, a quantum computer would need an efficiency better than $3.49 \times 10^5 \times 10^{-10} \approx 10$ $\mu$J per Grover iteration to be advantageous over the current mining efficiency of classical miners.  
\section{Discussion}

We analyze the feasibility of quantum Bitcoin mining for the scenario in which a single quantum miner competes in an otherwise classical network. We give a closed-form expression for the quantum miner's success when the quantum miner has small  computational power compared to the network. We then describe an optimal mining protocol for when the quantum miner is also peaceful. We analyze the quantum miner's efficiency and effective hash rate under these assumptions. Lastly, we give conditions for advantageous quantum mining which can be used to evaluate if a given quantum computer can provide an advantage over classical computers at mining. 

Although the average time for blocks to be mined by the network is $\lambda^{-1}= 10$ minutes, we find that, surprisingly, a quantum miner should wait until $y_0/\lambda_0\approx 16$ minutes to measure. The reason they should measure at this time is that the low probability $(20\%)$ of reaching the measurement step is counteracted by an increased probability that measurement yields a block. Such an effect is directly caused by the superlinear scaling of the quantum miner's success probability with respect to the number of Grover iterations applied. 

The computational problem of quantum Bitcoin mining contains an embedded time limit in that the quantum miner race to find a block before any other miner. This limit yields slightly non-intuitive dependencies on problem parameters. These dependencies are clear in the small computational power regime. Most noticeably, the probability of success for a peaceful quantum miner scales approximately linearly with problem difficulty $D$. This is the same scaling as classical mining and so runs counter to the usual intuition that quantum advantage increases with harder and harder problems. The reason quantum advantage does not increase with difficulty  is  the time limitation on mining which forces the quantum miner to apply only a fraction of the complete algorithm for Grover's search. As problem difficulty increases, this fraction becomes smaller as the total number of Grover iterations applied for each block stays constant, but the number of iterations to required to complete Grover's algorithm increases.

Instead, quantum advantage grows as the speed of the quantum computer $r$ increases because as $r$ becomes larger, the quantum miner is able to apply a larger fraction of Grover's algorithm. This scaling is due to the classical mining success probability having only a linear dependence on gate speed, whereas quantum success probability has an approximately quadratic dependence on gate speed. Finally, we note that the quantum advantage also increases as $\lambda^{-1}$ increases for similar reasons as discussed for the $r$ scaling; more time between blocks gives the quantum miner more time to apply Grover iterations. This scaling suggests that quantum mining is more effective if the blockchain protocol has longer times $\lambda^{-1}$ between blocks.

We now compare our effective hash rate calculation to that in \cite{Aggarwal}.  Aggarwal et al. aim to evaluate when and if a quantum miner could dominate the Bitcoin networks hashing power. For this reason, they assume that the quantum miner applies the entirety of Grover's algorithm, as is needed for the quantum miner to succeed with probability close to 1. They calculate that, with $r=66.7 $MHz , $D \approx .42 \times 10^{20} $, and a factor of $538.6$ increase in the number of gates per Grover iteration due to error correction (total gates per Grover iteration is $538.6 \cdot 297784$), a quantum miner could achieve an effective hash rate of $13\text{GH/s}$. This effective hash rate is much larger than the effective hash rate of $78 \text{MH/s}$ we calculate for a computer with no error correction overhead that is otherwise identical at the same difficulty $D$.

The discrepancy comes from two places. First, we are interested in the regime in which the quantum miner has small computational power compared to the network. The $r=66.7 $MHz quantum computer satisfies this regime well. Under these assumptions the entirety of Grover's algorithm can not be completed. Therefore less quadratic speed-up can be reaped from the quadratic scaling difference, and the quantum miner performs worse than they would if we assumed Grover's algorithm was completed (as is assumed in \cite{Aggarwal}). Second, the effective hash rate calculation by Aggarwal et al. does not account for Grover iterations that are wasted because a classical miner finds a block before measurement is reached.  Again, this approximation fits the regime they consider, as a quantum miner who is always able to complete Grover's algorithm does not waste Grover iterations. 

The fundamentals of Grover search are well understood. However,
interesting properties of quantum search arise in the context of specific applications. We illustrate these properties in the case of a single quantum computer Bitcoin mining in an otherwise classical network. In order to derive an optimal quantum mining procedure we assume that the quantum computer can not dominate the network and  that no aggressive mining occurs. For future work, it would be interesting to see if an optimal quantum mining protocol can be determined when these assumptions are relaxed. It would also be valuable to numerically verify that our approximation $\tilde P$ for the quantum miner's success probability $P$ agrees in the small computational power regime.  

\section*{Acknowledgements} The authors thank Shahadat Hossain and Robert Benkoczi for helpful discussions throughout the development of this work and thank Barry Sanders for insightful comments on the manuscript. This work was supported by the Alberta Government We acknowledge the traditional owners of the land on which this work was undertaken at the University of Calgary and the University of Lethbridge: the Treaty 7 First Nations \href{www.treaty7.org}{www.treaty7.org}.

\bibliography{references}{}
\bibliographystyle{abbrv}

\end{document}